\documentclass[sigconf]{acmart}
\usepackage{bm}
\usepackage{upgreek}
\usepackage[linesnumbered, ruled, vlined]{algorithm2e}
\usepackage{scalefnt}
\usepackage{setspace}
\usepackage{enumitem}
\usepackage{epstopdf}
\usepackage{dsfont}
\usepackage[english]{babel}
\usepackage{etoolbox}
\usepackage{amsmath}
\usepackage{empheq}
\usepackage{cases}
\usepackage{tikz}

\newcommand*\circled[1]{\tikz[baseline=(char.base)]{
            \node[shape=circle,draw,inner sep=0.25pt] (char) {#1};}}

\renewenvironment{thebibliography}[1]{%
 \begin{oldthebibliography}{#1}%
   \setlength{\parskip}{0ex}%
   \setlength{\itemsep}{0ex}%
}%
{%
  \end{oldthebibliography}%
}

\allowdisplaybreaks
\newtheorem{myDef}{Definition}
\newtheorem{myTheo}{Theorem}
\newtheorem{myLemma}{Lemma}
\newtheorem{myCor}{Corollary}

\begin{document}

\title{Theseus: Incentivizing Truth Discovery in Mobile Crowd Sensing Systems}
\titlenote{We gratefully acknowledge the support of National Science Foundation grants CNS-1330491, 1566374, and 1652503.}
\author{Haiming Jin}
\affiliation{%
  \institution{Department of Computer Science, University of Illinois at Urbana-Champaign, IL, USA}
}
\email{hjin8@illinois.edu}
\author{Lu Su}
\affiliation{%
  \institution{Department of Computer Science and Engineering, State University of New York at Buffalo, NY, USA}
}
\email{lusu@buffalo.edu}

\author{Klara Nahrstedt}
\affiliation{%
  \institution{Department of Computer Science, University of Illinois at Urbana-Champaign, IL, USA}
}
\email{klara@illinois.edu}

\begin{abstract}
The recent proliferation of human-carried mobile devices has given rise to mobile crowd sensing (MCS) systems that outsource sensory data collection to the public crowd. In order to identify truthful values from (crowd) workers' noisy or even conflicting sensory data, \textit{truth discovery algorithms}, which jointly estimate workers' data quality and the underlying truths through quality-aware data aggregation, have drawn significant attention. However, the power of these algorithms could not be fully unleashed in MCS systems, unless workers' \textit{strategic reduction of their sensing effort} is properly tackled. To address this issue, in this paper, we propose a \textit{payment mechanism}, named Theseus, that deals with workers' such strategic behavior, and incentivizes high-effort sensing from workers. We ensure that, at the \textit{Bayesian Nash Equilibrium} of the \textit{non-cooperative game} induced by Theseus, all participating workers will spend their \textit{maximum possible effort} on sensing, which improves their data quality. As a result, the aggregated results calculated subsequently by truth discovery algorithms based on workers' data will be highly accurate. Additionally, Theseus bears other desirable properties, including \textit{individual rationality} and \textit{budget feasibility}. We validate the desirable properties of Theseus through theoretical analysis, as well as extensive simulations.
\end{abstract}





\maketitle

\section{Introduction}

The recent proliferation of increasingly capable human-carried mobile devices (e.g., smartphones, smartwatches, smartglasses) equipped with a plethora of on-board sensors (e.g., accelerometer, compass, gyroscope, GPS, camera) has given rise to mobile crowd sensing (MCS), a new sensing paradigm which outsources sensory data collection to a crowd of participants, namely (crowd) workers. Thus far, a wide spectrum of MCS systems \cite{ShaohanTOSN15, YunSenSys14, JakobMobiSys08, RaghuMobiSys10, MTMC16} have been deployed which cover almost every aspect of our lives, including smart transportation, healthcare, environmental monitoring, indoor localization, and many others. 

In real practice, workers' sensory data are usually unreliable because of various factors (e.g., lack of effort, insufficient skill, poor sensor quality, background noise). Thus, the crowd sensing platform, which is usually a cloud-based central server, has to properly aggregate workers' noisy or even conflicting data so as to obtain accurate aggregated results. Clearly, a \textit{weighted aggregation method} that assigns higher weights to workers with more reliable data is much more favorable than naive methods (e.g., averaging and voting) that view each worker equally, in that it shifts the aggregated results towards the data provided by more reliable workers. 

The challenge, however, is that workers' reliability is usually unknown \textit{a priori} by the platform, and should be inferred from the sensory data submitted by individual workers. To address this issue, \textit{truth discovery}, which refers to a family of algorithms \cite{QiVLDB14, QiSIGMOD14, YaliangKDD15, XiaoxinTKDE08} that aim to discover meaningful facts from unreliable data, has been proposed and widely studied. Without any prior knowledge about workers' reliability, a truth discovery algorithm calculates jointly workers' weights and the aggregated results, based on the principles that the workers whose data are closer to the aggregated results will be assigned higher weights, and the data from a worker with a higher weight will be counted more in the aggregation. 

Though yielding reasonably good performance under certain circumstances, truth discovery algorithms still suffer from the limitation that the aggregation accuracy highly depends on the quality of input data. If a vast majority of the data sources are unreliable, it will be hard or even impossible for these algorithms to obtain accurate aggregated results. This is exactly why past literature on truth discovery \cite{QiVLDB14, QiSIGMOD14, YaliangKDD15, XiaoxinTKDE08} assumes that most data sources have fairly good reliability. However, in MCS systems, such assumption does not hold, as the data sources here are \textit{selfish} workers, who may \textit{strategically reduce their costly sensing effort}, such as the time, resources, attention, and carefulness they put into the sensing tasks. Clearly, the level of a worker's sensing effort is among the major factors that affect her data quality. The reduction of workers' effort inevitably deteriorates the quality of their sensory data, which further impairs the aggregation accuracy. For example, in air quality monitoring applications \cite{YunSenSys14}, in order to save effort, workers may carry their mobile devices in their pockets instead of holding them on their hands as required, which may significantly degrade the reliability of their air quality measurements. Therefore, \textit{the power of truth discovery algorithms could not be fully unleashed in MCS systems, unless the platform properly deals with workers' strategic reduction of sensing effort}.

To address this issue, in this paper, we take into consideration workers' strategic behavior, and propose a \textit{payment mechanism}, named Theseus\footnote{The name Theseus comes from incen\underline{T}ivizing trut\underline{h} discov\underline{e}ry with \underline{s}trat\underline{e}gic data so\underline{u}rce\underline{s}.}, that offers payments to incentivize high-effort sensing from workers. Our workflow of an MCS system starts with the platform announcing the Theseus payment mechanism to workers before all the sensing happens. Workers' strategic behavior after the announcement of Theseus is then modeled using \textit{game-theoretic methods}. In our model, Theseus induces a \textit{non-cooperative game}\footnote{\label{def:ncgame}Non-cooperative game refers to the family of games, where each player acts independently without collaboration or communication with others, whereas, in cooperative games, players may communicate with each other and form coalitions.}, called \textit{sensing game}, where workers are the players who strategically decide their levels of effort for sensing. In order to elicit effort from workers, Theseus is then designed such that at the \textit{Bayesian Nash Equilibrium (BNE)} of the sensing game, each participating worker maximizes her expected utility only when she spends her maximum possible effort. Clearly, Theseus improves the quality of workers' data by controlling a critical factor, that is, the level of their sensing effort. As a result, the aggregated results calculated subsequently by truth discovery algorithms based on workers' sensory data will be of high accuracy. 

In summary, this paper makes the following contributions. 
\begin{itemize}[leftmargin=*]
\item In this paper, we propose a \textit{payment mechanism}, called Theseus, which is used \textit{in pair with a truth discovery algorithm} to ensure \textit{high aggregation accuracy} in MCS systems where workers may strategically reduce their effort for sensing. 
\item Our Theseus payment mechanism deals with workers' strategic behavior by incentivizing workers to \textit{spend their maximum possible sensing effort} at the BNE of the induced non-cooperative game among them. 
\item Additionally, we ensure that Theseus bears other desirable properties, including \textit{individual rationality} and \textit{budget feasibility}. 
\end{itemize}

\section{Related Work}\label{sec:related}

In order to identify truthful values from workers' noisy or even conflicting sensory data in MCS systems, truth discovery algorithms \cite{QiVLDB14, QiSIGMOD14, YaliangKDD15, XiaoxinTKDE08}, which jointly estimate workers' data quality and the underlying truths through quality-aware data aggregation, have drawn significant attention. However, \textit{these algorithms usually cannot deal with workers' strategic reduction of sensing effort}, and thus, may yield unsatisfactory aggregation accuracy. 

Another line of prior work related to this paper is a series of incentive mechanisms \cite{IordanisINFOCOM13, RadanovicAAMAS15, RadanovicTIST16, LingjieINFOCOM12, ManHonMobiHoc15, TieINFOCOM15, YanjiaoINFOCOM16, DejunMobiCom12, HaimingINFOCOM17, ZhengniINFOCOM14, YuemingINFOCOM15, XinglinTPDS14, QiINFOCOM15, HaimingMobiHoc15, HaimingMobiHoc16, DongINFOCOM14, DongTON16, LinINFOCOM15, XiangINFOCOM15, HonggangINFOCOM16, JingICDCS16, YutianTVT15, WeinaSIGMETRICS16, WeinaAllerton15, StratisWINE13, RachelCOLT15, KaiMobiHoc16, ShuoIWQoS16, MerkouriosMobiHoc16, MerkouriosINFOCOM15, ShiboINFOCOM14, LingjunINFOCOM16, DanMobiHoc15} recently developed by the research community in order to stimulate worker participation in MCS systems. Most of these past literature \cite{IordanisINFOCOM13, RadanovicAAMAS15, RadanovicTIST16, LingjieINFOCOM12, ManHonMobiHoc15, TieINFOCOM15, YanjiaoINFOCOM16, DejunMobiCom12, HaimingINFOCOM17, ZhengniINFOCOM14, YuemingINFOCOM15, XinglinTPDS14, QiINFOCOM15, HaimingMobiHoc15, HaimingMobiHoc16, DongINFOCOM14, DongTON16, LinINFOCOM15, XiangINFOCOM15, HonggangINFOCOM16, JingICDCS16, YutianTVT15, WeinaSIGMETRICS16, WeinaAllerton15, StratisWINE13, RachelCOLT15, KaiMobiHoc16} adopts game-theoretic methods, due to their ability to deal with workers' strategic behavior. Among them, auction-based incentive mechanisms \cite{DejunMobiCom12, HaimingINFOCOM17, ZhengniINFOCOM14, YuemingINFOCOM15, XinglinTPDS14, QiINFOCOM15, HaimingMobiHoc15, HaimingMobiHoc16, DongINFOCOM14, DongTON16, LinINFOCOM15, XiangINFOCOM15, HonggangINFOCOM16, JingICDCS16, YutianTVT15} typically consider workers' strategic bidding of the prices and sensing task choices to the platform. Furthermore, some prior work \cite{WeinaSIGMETRICS16, WeinaAllerton15, StratisWINE13, RachelCOLT15} tackles workers' strategic manipulation of reported private and sensitive data due to privacy concerns. However, \textit{none of them study workers' strategic reduction of sensing effort as in this work}. Mechanisms that elicit effort from crowd workers have been investigated in past literature \cite{IordanisINFOCOM13, RadanovicAAMAS15, RadanovicTIST16, LingjieINFOCOM12, YanjiaoINFOCOM16, DejunMobiCom12, ManHonMobiHoc15,TieINFOCOM15}, but \textit{none of them is designed to work in pair with truth discovery algorithms}. Note that although one existing incentive mechanism \cite{DanMobiHoc15} is able to work jointly with truth discovery algorithms, it is not based on game-theoretic models, and thus, cannot tackle workers' strategic behavior. 

Different from existing work, in this paper, we design a payment mechanism, which is used in pair with a truth discovery algorithm to ensure high aggregation accuracy by incentivizing workers to spend their maximum possible sensing effort.

\section{Preliminaries}
In this section, we introduce the system overview, truth discovery algorithms, our game theoretic model, as well as the design objectives. 
\subsection{System Overview}\label{sec:overview}
We consider an MCS system consisting of a cloud-based platform, and a set of $S$ potential participating workers, denoted as $\mathcal{S}=\{1,2\cdots,S\}$. The platform holds a set of $M$ sensing tasks, denoted as $\mathcal{M}=\{1,2,\cdots,M\}$, and each task requires workers to sense a particular object, event, or phenomenon locally, and report to the platform the sensory data in the form of \textit{continuous values}. Such MCS systems collecting continuous data from the crowd, constitute a large portion of the currently deployed MCS systems, such as environmental monitoring applications that collect air quality or noise level measurements from participating workers. We demonstrate the interaction between the platform and workers in Figure \ref{fig:sysframe}, and describe the complete workflow of our MCS system model as follows. 

\begin{figure}[htb]
\centering
\includegraphics[width=0.4\textwidth]{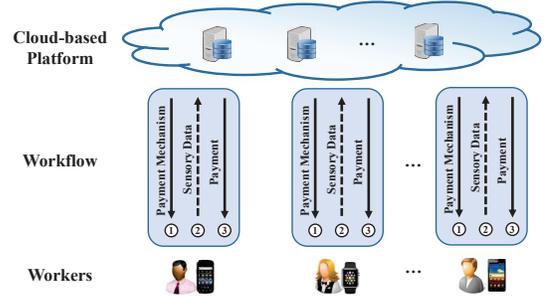}
\caption{Interaction between the platform and workers (where circled numbers represent the order of the events).}
\label{fig:sysframe}
\end{figure}

\begin{itemize}[leftmargin=*]
\item Firstly, the platform announces the set of sensing tasks $\mathcal{M}$, as well as the payment mechanism,  to the set of all potential participating workers $\mathcal{S}$ (step \circled{1}). 
\item After such announcements, each worker $s\in\mathcal{S}$ decides whether or not to participate in the sensing tasks. Then, the workers that choose to participate decide the levels of their sensing effort (e.g., time, resources, attention, carefulness), and carry out sensing according to the decided effort levels. We denote the set of participating workers as $\mathcal{S}'\subseteq\mathcal{S}$. Each worker $s\in\mathcal{S}'$ then submits to the platform the sensory data $x_m^s$ for each task $m\in\mathcal{M}$ upon completion of sensing\footnote{Clearly, in practice, each individual worker may not be able to execute all the sensing tasks hosted by the platform. Thus, a more realistic model is to introduce an affinity term for each worker-task pair $(s,m)$ that indicates whether or not worker $s$ is able to execute task $m$. However, to simplify the presentation of our subsequent mathematical analyses, we assume that each worker is capable to execute all the tasks.} (step \circled{2}).
\item After receiving workers' data, the platform pays each participating worker according to the payment calculated using the payment mechanism (step \circled{3}).
\item Finally, based on the collected data, the platform calculates an aggregated result $x_m^*$ for each task $m$, and uses it as an estimate for the ground truth $x_m^{\text{truth}}$, which is unknown to both the platform and the workers. 
\end{itemize}

As the quality of different workers' sensory data typically varies, an ideal approach is to use a weighted aggregation scheme which assigns higher weights to workers with higher data quality. However, in practice, workers' data  quality is usually unknown \textit{a priori} to the platform. Therefore, in our model, the platform utilizes one of the \textit{truth discovery} algorithms \cite{QiVLDB14, QiSIGMOD14, YaliangKDD15, XiaoxinTKDE08} to aggregate workers' data, which calculates workers' weights and estimates the ground truths in a joint manner. An introduction of such algorithms is provided in the following Section \ref{sec:td}. 

\subsection{Truth Discovery}\label{sec:td}
Although existing truth discovery algorithms \cite{QiVLDB14, QiSIGMOD14, YaliangKDD15, XiaoxinTKDE08} differ in their specific ways to calculate workers' weights and the aggregated results, their common procedure could be summarized as in the following Algorithm \ref{al:td}. 

\begin{algorithm}[h]
\small
\KwIn{Workers' data $\{x_m^s|m\in\mathcal{M}, s\in\mathcal{S}'\}$\;}
\KwOut{Estimated ground truths $\{x_m^*|m\in\mathcal{M}\}$\;}
Randomly initialize the ground truth for each task\;
\Repeat{Convergence criterion is satisfied}	
{
	\tcp{Weight calculation}
	\ForEach{$s\in\mathcal{S}'$}{
		Update the weight $w_s$ based on current estimated ground truths using Equation (\ref{eq:weight})\;
	}	
	\tcp{Truth estimation}
 	\ForEach{$m\in\mathcal{M}$}{
		Update the estimated ground truth $x_m^*$ based on workers' current weights using Equation (\ref{eq:truth})\;
	}	
}
\Return Estimated ground truths $\{x_m^*|m\in\mathcal{M}\}$\;
\caption{Truth Discovery Algorithm}\label{al:td}
\end{algorithm}

A truth discovery algorithm, as described in Algorithm \ref{al:td}, typically starts with a random guess of tasks' ground truths, and then iteratively updates workers' weights, as well as the estimated ground truths until convergence.

\textbf{Weight Calculation}. In this step, tasks' estimated ground truths are assumed to be fixed, and the weight $w_s$ of each worker $s\in\mathcal{S}'$ is calculated as
\begin{equation}\label{eq:weight}
\begin{aligned}
w_s=\omega\Bigg(\sum_{m\in\mathcal{M}}d(x_m^s,x_m^*)\Bigg),
\end{aligned}
\end{equation}
where $\omega(\cdot)$ is some monotonically decreasing function, and $d(\cdot)$ denotes the function that calculates the distance between the worker's data $x_m^s$ and the estimated ground truth $x_m^*$. Although different truth discovery algorithms may adopt different functions $\omega(\cdot)$ and $d(\cdot)$, they share the same underlying principle that higher weights are assigned to workers whose data are closer to the estimated ground truths. 

\textbf{Truth Estimation}. In this step, workers' weights are assumed to be fixed, and the estimated ground truth $x_m^*$ of each task $m$ is derived as 
\begin{equation}\label{eq:truth}
\begin{aligned}
x_m^*=\frac{\sum_{s\in\mathcal{S}'}w_s x_m^s}{\sum_{s\in\mathcal{S}'}w_s}.
\end{aligned}
\end{equation}

In such weighted aggregation method, the aggregated result $x_m^*$ relies more on the workers with higher weights. Usually, the convergence criterion is application specific. For example, the algorithm could be treated as converged as long as the difference between the estimated ground truths in two consecutive iterations is less than a threshold.

Note that the payment mechanism that we propose in this paper is independent with the specific forms of the functions $\omega(\cdot)$ and $d(\cdot)$ in Equation (\ref{eq:weight}). Therefore, it is able to work jointly with any truth discovery algorithm that shares the same procedure as Algorithm \ref{al:td}. Further discussions on this point will be provided in Section \ref{sec:formulation}.

\subsection{Game Theoretic Model}\label{sec:gamemodel}
As the aggregation accuracy of truth discovery algorithms highly depends on the quality of input data, existing work on truth discovery \cite{QiVLDB14, QiSIGMOD14, YaliangKDD15, XiaoxinTKDE08} assumes that most data sources have fairly good reliability. In MCS systems, however, such assumption does not hold, as the data sources here are usually \textit{strategic} and \textit{selfish} workers, who may reduce their sensing effort strategically, and thus, provide unreliable data. 

In this paper, we take into consideration workers' strategic behavior, and incentivize workers to provide high quality data using a \textit{payment mechanism} defined in Definition \ref{def:paymentmechanism}.

\begin{myDef}[Payment Mechanism]\label{def:paymentmechanism} A payment mechanism, denoted as $p:\mathcal{X}\rightarrow\mathbb{R}^S$, where $\mathcal{X}$ denotes the set containing all possible sets of workers' sensory data, calculates the payments to workers based on the collected set of data $\mathbf{x}=\{x_m^s|m\in\mathcal{M}, s\in\mathcal{S}'\}$. We use $p_s(\mathbf{x})\geq 0$ to denote the payment to worker $s$, when the set of collected data is $\mathbf{x}$. Note that $p_s(\mathbf{x})=0$, if worker $s$ drops out. 
\end{myDef}

As mentioned in Section \ref{sec:overview}, the platform firstly announces to workers the payment mechanism $p(\cdot)$, which then induces a \textit{non-cooperative game}\footnote{Refer to Footnote \ref{def:ncgame} for definition.}, referred to as \textit{sensing game} in the rest of this paper, where workers are the players. In this game, each worker decides whether or not to participate by evaluating her own expected utility. That is, a worker $s$ will drop out, if participation leads to a negative expected utility, and otherwise, she will participate with a specific effort level $e_s$ that maximizes her expected utility. Similar to past literature \cite{BoQDB12, QiVLDB14}, we assume that the difference between any worker $s$'s data and the ground truth follows a zero-mean Gaussian distribution, i.e., 
\begin{equation}\label{eq:gd}
\begin{aligned}
X_m^s-X_m^{\text{truth}}\sim N(0,\delta_s^2),
\end{aligned}
\end{equation}
where $X_m^s$ and $X_m^{\text{truth}}$ are the random variables corresponding to $x_m^s$ and the ground truth $x_m^{\text{truth}}$ respectively, and $N(0,\delta_s^2)$ denotes a Gaussian random variable with mean zero and standard deviation $\delta_s$. Although we assume that such difference follows a Gaussian distribution, the results in this paper could be generalized, with some adaptation, to scenarios with other types of distributions. Clearly, the standard deviation $\delta_s$ captures a worker $s$'s data quality, as the less the value of $\delta_s$, the more likely that her sensory data will be close to the ground truth.

As a worker's data quality typically increases with her effort level, we assume that $\delta_s=q_s(e_s)\in[\underline{\delta}_s, \overline{\delta}_s]$ for each worker $s$, where $q_s(\cdot)$ is a bounded monotonically non-increasing function. We allow, in our model, workers to have different $q_s(\cdot)$ functions and ranges for their $\delta_s$'s, because apart from a worker's effort, her data quality is also affected by other factors (e.g., skill level, sensor quality, environment noise). As each worker $s$ is assigned a single weight $w_s$ in the truth discovery algorithm adopted by us (Algorithm \ref{al:td}), we assume that she spends the same amount of effort $e_s$ on all the tasks. We leave the study of the scenario where workers have different effort levels on different tasks in our future work.

For simplicity, we use $\delta_s$ instead of $e_s$ as a worker $s$'s \textit{strategy}, and use $\delta_s=\bot$ to denote that the worker chooses to drop out. Thus, a worker $s$'s \textit{strategy space} is $[\underline{\delta}_s, \overline{\delta}_s]\cup\{\bot\}$. As given by Equation (\ref{eq:gd}), the distribution of any worker $s$'s data depends on $\delta_s$, we use $\mathbf{x}(\boldsymbol{\updelta})$ to denote the set of collected data, and $\mathbf{X}(\boldsymbol{\updelta})$ the random variable corresponding to $\mathbf{x}(\boldsymbol{\updelta})$, when workers' \textit{strategy profile} is $\boldsymbol{\updelta}=(\delta_1,\delta_2,\cdots,\delta_S)$. Then, we define a worker's utility in Definition \ref{def:utility}. 

\begin{myDef}[Worker's Utility]\label{def:utility} Given the payment mechanism $p(\cdot)$ and workers' strategy profile $\boldsymbol{\updelta}=(\delta_1,\delta_2,\cdots,\delta_S)$, any worker $s$'s utility is
\begin{equation}
\begin{aligned}
u_s(\boldsymbol{\updelta})=p_s\big(\mathbf{x}(\boldsymbol{\updelta})\big)-C_s(\delta_s),
\end{aligned}
\end{equation}
where $C_s(\cdot)$ is a monotonically decreasing function for $\delta_s\in[\underline{\delta}_s, \overline{\delta}_s]$, and $C_s(\bot)=0$. $C_s(\delta_s)$ denotes worker $s$'s sensing cost when her strategy is $\delta_s$. Therefore, the expected utility of worker $s$ (evaluated by worker $s$) is
\begin{equation}\label{eq:expectedutility}
\begin{aligned}\mathbb{E}_{\boldsymbol{\updelta}_{-s}}\big[u_s(\delta_s,\boldsymbol{\updelta}_{-s})\big]=\mathbb{E}_{\boldsymbol{\updelta}_{-s}}\big[p_s\big(\mathbf{X}(\delta_s,\boldsymbol{\updelta}_{-s})\big)\big]-C_s(\delta_s),
\end{aligned}
\end{equation}
where $\boldsymbol{\updelta}_{-s}=(\delta_1,\cdots,\delta_{s-1},\delta_{s+1},\cdots,\delta_S)$ denotes workers' strategy profile excluding $\delta_s$. 
\end{myDef}
In general cases, the calculation of a worker $s$'s expected utility in Equation (\ref{eq:expectedutility}) requires the knowledge of the joint distribution of $\boldsymbol{\updelta}_{-s}$. However, because of the specific design of our payment mechanism described in Section \ref{sec:mechanism}, the calculation can be done without knowing such joint distribution. We leave the detailed discussion on the required prior statistical knowledge in Section \ref{sec:Parameterization}.

\subsection{Design Objectives}

In this paper, we aim to design a payment mechanism which preserves several desirable properties at the \textit{Bayesian Nash Equilibrium (BNE)}, formally defined in Definition \ref{def:BNE}, of the sensing game. 

\begin{myDef}[BNE]\label{def:BNE} The strategy profile $\boldsymbol{\updelta}^*=(\delta_1^*,\delta_2^*,\cdots,\delta_S^*)$ is a Bayesian Nash Equilibrium (BNE) of the sensing game, if
\begin{equation}
\begin{aligned}
\mathbb{E}_{\boldsymbol{\updelta}_{-s}^*}\big[u_s(\delta_s^*,\boldsymbol{\updelta}_{-s}^*)\big]\geq\mathbb{E}_{\boldsymbol{\updelta}_{-s}^*}\big[u_s(\delta_s,\boldsymbol{\updelta}_{-s}^*)\big], \forall s\in\mathcal{S}, \delta_s,
\end{aligned}
\end{equation}
where $\boldsymbol{\updelta}_{-s}^*=(\delta_1^*,\cdots,\delta_{s-1}^*,\delta_{s+1}^*,\cdots,\delta_S^*)$.
\end{myDef}

Clearly, BNE $\boldsymbol{\updelta}^*$ satisfies that any worker $s$ maximizes her expected utility by taking strategy $\delta_s^*$ given that other workers take strategies $\boldsymbol{\updelta}_{-s}^*$. One desirable property we aim to achieve is \textit{individual rationality} defined in Definition \ref{def:ir}. 

\begin{myDef}[Individual Rationality]\label{def:ir} A payment mechanism $p(\cdot)$ is individual rational, if and only if no worker has negative expected utility at BNE $\boldsymbol{\updelta}^*$, i.e.,
\begin{equation}
\begin{aligned}
\mathbb{E}_{\boldsymbol{\updelta}_{-s}^*}\big[u_s(\delta_s^*,\boldsymbol{\updelta}_{-s}^*)\big]\geq 0, \forall s\in\mathcal{S}.
\end{aligned}
\end{equation}
\end{myDef}

The property of individual rationality is necessary for a payment mechanism, as it prevents workers from being disincentivized to participate. Because usually, in practice, the platform works under a fixed budget, another design objective considered is \textit{budget feasibility} defined in Definition \ref{def:budget}. 

\begin{myDef}[Budget Feasibility]\label{def:budget} A payment mechanism $p(\cdot)$ is budget feasible, if and only if the expected overall payment at BNE $\boldsymbol{\updelta}^*$ does not exceed the budget $B$, i.e., 
\begin{equation}
\begin{aligned}
\mathbb{E}_{\boldsymbol{\updelta}^*}\bigg[\sum_{s\in\mathcal{S}'}p_s\big(\mathbf{X}(\boldsymbol{\updelta}^*)\big)\bigg]\leq B.
\end{aligned}
\end{equation}
\end{myDef}

Another critical desirable property is that workers at BNE provide high quality data, so that the truth discovery algorithm ensures low \textit{error probability}, which is defined in Definition \ref{def:ep}.

\begin{myDef}[Error Probability]\label{def:ep} Given any $\alpha>0$, we define the error probability of a truth discovery algorithm as 
\begin{equation}
\begin{aligned}
\textnormal{Pr}\Bigg(\frac{1}{M}\sum_{m=1}^M\big|X_m^*-X_m^{\textnormal{truth}}\big|\geq\alpha\Bigg),
\end{aligned}
\end{equation}
where $X_m^*$ denotes the random variable corresponding to the estimated ground truth $x_m^*$. Clearly, it is the probability that the \textit{mean absolute error (MAE)}, $\frac{1}{M}\sum_{m=1}^M\big|X_m^*-X_m^{\textnormal{truth}}\big|$, of a truth discovery algorithm is no less than a given threshold $\alpha$. 
\end{myDef}

In summary, our objective is to design an \textit{individual rational} and \textit{budget feasible payment mechanism}, which ensures that the truth discovery algorithm guarantees \textit{low error probability at BNE}. 

\section{Mathematical Formulation}\label{sec:formulation}

In this section, we formally formulate the payment mechanism design problem mathematically. Firstly, we introduce the following Lemma \ref{theo:ub} that establishes an upper bound for the error probability of a truth discovery algorithm defined in Definition \ref{def:ep}. 

\begin{myLemma}\label{theo:ub}
Given any $\alpha>0$ and workers' strategy profile $\boldsymbol{\updelta}=(\delta_1,\delta_2,\cdots,\delta_S)$, we have that
\begin{equation}\label{eq:ub}
\begin{aligned}
\textnormal{Pr}\Bigg(\frac{1}{M}\sum_{m=1}^M\big|X_m^*-X_m^{\textnormal{truth}}\big|\geq\alpha\Bigg)\leq\sqrt{\frac{2}{\pi}}\frac{\sum_{s\in\mathcal{S}'}\delta_s}{\alpha}, 
\end{aligned}
\end{equation}
that is, the error probability of a truth discovery algorithm is upper bounded by $\sqrt{\frac{2}{\pi}}\frac{\sum_{s\in\mathcal{S}'}\delta_s}{\alpha}$.
\end{myLemma}

\begin{proof}
The MAE of a truth discovery algorithm satisfies that 
\begin{equation*}
\begin{aligned}
\frac{1}{M}\sum_{m=1}^M\big|X_m^*-X_m^{\textnormal{truth}}\big|&=\frac{1}{M}\sum_{m=1}^M\bigg|\frac{\sum_{s\in\mathcal{S}'}w_s X_m^s}{\sum_{s\in\mathcal{S}'}w_s}-X_m^{\textnormal{truth}}\bigg|\\
&=\frac{1}{M}\sum_{m=1}^M\Bigg|\frac{\sum_{s\in\mathcal{S}'}w_s\big(X_m^s-X_m^{\textnormal{truth}}\big)}{\sum_{s\in\mathcal{S}'}w_s}\Bigg|\\
&\leq\frac{1}{M}\frac{\sum_{m=1}^M\sum_{s\in\mathcal{S}'}w_s\big|X_m^s-X_m^{\textnormal{truth}}\big|}{\sum_{s\in\mathcal{S}'}w_s}\\
&=\frac{1}{M}\frac{\sum_{s\in\mathcal{S}'}w_s\Big(\sum_{m=1}^M\big|X_m^s-X_m^{\textnormal{truth}}\big|\Big)}{\sum_{s\in\mathcal{S}'}w_s}\\
&\leq\sum_{s\in\mathcal{S}'}\frac{1}{M}\sum_{m=1}^M\big|X_m^s-X_m^{\textnormal{truth}}\big|.
\end{aligned}
\end{equation*}

As $X_m^s-X_m^{\text{truth}}\sim N(0,\delta_s^2)$, we have that $\mathbb{E}\big[\big|X_m^s-X_m^{\textnormal{truth}}\big|\big]=\sqrt{\frac{2}{\pi}}\delta_s$. Thus, given any $\alpha>0$, we have that 

\begin{equation*}
\scalefont{0.93}
\begin{aligned}
\textnormal{Pr}\Bigg(\frac{1}{M}\sum_{m=1}^M\big|X_m^*-X_m^{\textnormal{truth}}\big|\geq\alpha\Bigg)&\leq\textnormal{Pr}\Bigg(\sum_{s\in\mathcal{S}'}\frac{1}{M}\sum_{m=1}^M\big|X_m^s-X_m^{\textnormal{truth}}\big|\geq\alpha\Bigg)\\
\text{(Markov's Inequality)}&\leq\frac{\mathbb{E}\big[\sum_{s\in\mathcal{S}'}\frac{1}{M}\sum_{m=1}^M\big|X_m^s-X_m^{\textnormal{truth}}\big|\big]}{\alpha}\\
&=\frac{\sum_{s\in\mathcal{S}'}\frac{1}{M}\sum_{m=1}^M\mathbb{E}\big[\big|X_m^s-X_m^{\textnormal{truth}}\big|\big]}{\alpha}\\
&=\sqrt{\frac{2}{\pi}}\frac{\sum_{s\in\mathcal{S}'}\delta_s}{\alpha},
\end{aligned}
\end{equation*}
which is exactly Inequality (\ref{eq:ub}). 
\end{proof}

Given any fixed $\alpha$, the upper bound of the error probability of a truth discovery algorithm given by Lemma \ref{theo:ub} is proportional to $\sum_{s\in\mathcal{S}'}\delta_s$, i.e., the sum of all participating workers' $\delta_s$'s. Thus, we aim to minimize $\sum_{s\in\mathcal{S}'}\delta_s^*$ in order to get a good guarantee for the error probability at BNE $\boldsymbol{\updelta}^*=(\delta_1^*,\delta_2^*,\cdots,\delta_S^*)$. The formal mathematical formulation of the \textit{payment mechanism design (PMD)} problem is given in the following optimization program. 

\textbf{PMD Problem:}
\begin{align}
\min_{p(\cdot)\in\mathcal{P}}&\sum_{s\in\mathcal{S}'}\delta_s^*\label{eq:obj}\\
\text{s.t.}&~\mathbb{E}_{\boldsymbol{\updelta}_{-s}^*}\big[u_s(\delta_s^*,\boldsymbol{\updelta}_{-s}^*)\big]\geq 0,~\forall s\in\mathcal{S}\label{eq:ir}\\
&~\mathbb{E}_{\boldsymbol{\updelta}^*}\bigg[\sum_{s\in\mathcal{S}'}p_s\big(\mathbf{X}(\boldsymbol{\updelta}^*)\big)\bigg]\leq B\label{eq:budget}
\end{align}

\textbf{Constants.} The PMD problem takes as inputs the worker set $\mathcal{S}$, the budget $B$, as well as the set $\mathcal{P}$, which denotes the set consisting of all the possible payment mechanisms, such that a BNE exists for the corresponding sensing game.

\textbf{Variable.} The variable of the PMD problem is the payment mechanism $p(\cdot)$. Furthermore, $\boldsymbol{\updelta}^*$ denotes the BNE corresponding to $p(\cdot)$, and $\mathcal{S}'$ and $\mathbf{X}(\boldsymbol{\updelta}^*)$ denote, respectively, the set of participating workers and collected sensory data at the BNE $\boldsymbol{\updelta}^*$. Note that as $\boldsymbol{\updelta}^*$, $\delta_s^*$, $\boldsymbol{\updelta}_{-s}^*$, and $\mathcal{S}'$ in the PMD problem are determined by $p(\cdot)$, more comprehensive notations of them are $\boldsymbol{\updelta}^*(p(\cdot))$, $\delta_s^*(p(\cdot))$, $\boldsymbol{\updelta}_{-s}^*(p(\cdot))$, and $\mathcal{S}'(p(\cdot))$, respectively. For simplicity, however, we denote them as $\boldsymbol{\updelta}^*$, $\delta_s^*$, $\boldsymbol{\updelta}_{-s}^*$, and $\mathcal{S}'$ as in the PMD problem.

\textbf{Objective function.} The objective (Equation (\ref{eq:obj})) of the PMD problem is to find the payment mechanism from $\mathcal{P}$ with the minimum $\sum_{s\in\mathcal{S}'}\delta_s^*$ at the corresponding BNE $\boldsymbol{\updelta}^*$, which is equivalent to minimizing the upper bound, as derived in Lemma \ref{theo:ub}, of a truth discovery algorithm's error probability at BNE for a fixed $\alpha$.

\textbf{Constraints.} Constraint (\ref{eq:ir}) and (\ref{eq:budget}) ensure, respectively, that any feasible solution $p(\cdot)$ to the PMD problem satisfies individual rationality and budget feasibility. 

Thus, the PMD problem aims to find the individual rational and budget feasible payment mechanism, which minimizes the upper bound (given by Lemma \ref{theo:ub}) of a truth discovery algorithm's error probability at the corresponding BNE for any fixed $\alpha$. Clearly, our formulation of the PMD problem is valid for an arbitrary way of assigning workers' weights. Therefore, the above formulation and the proposed payment mechanism to be presented in the following section can be applied to any truth discovery algorithm that has the same procedure as Algorithm \ref{al:td}.

\section{Proposed Payment Mechanism}\label{sec:mechanism}

As solving directly the optimal payment mechanism is hard, in this section, we propose our own payment mechanism, named Theseus, in Algorithm \ref{al:theseus}, which approximately solves the PMD problem with good performance guarantees. 

\begin{algorithm}
\small
\KwIn{$\mathcal{M}$, $\mathcal{S}$, $\mathcal{S}'$, $\mathbf{x}$, $\{(a_s,b_s)|s\in\mathcal{S}\}$\;}
\KwOut{$\{p_s|s\in\mathcal{S}\}$\;}
\ForEach{\textnormal{worker} $s\in\mathcal{S}$}{
\eIf{$s\in\mathcal{S}'$}{\label{line:ifstart}
	Randomly pick another worker $r\in\mathcal{S}'$\;\label{line:pickr}
	$p_s\leftarrow b_s-a_s\frac{1}{M}\sum_{m=1}^M(x_m^s-x_m^{r})^2$\;\label{line:ifend}	
}{
	$p_s\leftarrow 0$\;\label{line:else}
}}
\Return$\{p_s|s\in\mathcal{S}\}$\;\label{line:returnwinner}
\caption{Theseus Payment Mechanism}\label{al:theseus}
\end{algorithm}

Algorithm \ref{al:theseus} takes as inputs the set of tasks $\mathcal{M}$, workers $\mathcal{S}$, and participating workers $\mathcal{S}'$, as well as the set of collected sensory data $\mathbf{x}$, and $\{(a_s,b_s)|s\in\mathcal{S}\}$ where $a_s$ and $b_s$ are positive parameters related to the payment to worker $s$. The calculation of the payment to any participating worker (line \ref{line:ifstart}-\ref{line:ifend}) borrows the high-level idea of the peer prediction  method \cite{NolanMS05}, which basically decides the payment based on the difference between her data and that of a randomly selected \textit{reference worker}. That is, if worker $s$ participates (i.e., $s\in\mathcal{S}'$), Algorithm \ref{al:theseus} randomly picks another reference worker $r$ from the set of participating workers $\mathcal{S}'$ (line \ref{line:pickr}). Next, the payment $p_s$ to this worker $s$ is set as 
\begin{equation}
\begin{aligned}
p_s=b_s-a_s\frac{1}{M}\sum_{m=1}^M(x_m^s-x_m^{r})^2. 
\end{aligned}
\end{equation}
Clearly, the more worker $s$'s data agrees with that of the randomly selected reference worker $r$, the higher her payment $p_s$ will be. If any worker $s$ drops out (i.e., $s\not\in\mathcal{S}'$), the algorithm will set her payment as 0 (line \ref{line:else}). Finally, the algorithm returns the set of payments to all workers $\{p_s|s\in\mathcal{S}\}$ (line \ref{line:returnwinner}). By now, our description of Theseus has been finished except for one missing piece, that is, how the parameters $\{(a_s,b_s)|s\in\mathcal{S}\}$ are set, which is presented in the following Section \ref{sec:Parameterization}. 

Clearly, another intuitive way of deciding the payment $p_s$ to each participating worker $s$ is to set $p_s$ to be positively correlated to her weight $w_s$ calculated by the truth discovery algorithm using Equation (\ref{eq:weight}). However, we do not adopt this approach due to the difficulty in analyzing the properties of the induced sensing game. 

\vspace{-0.15cm}
\section{Parameterization}\label{sec:Parameterization}
In this section, we introduce our careful selection of the parameters $\{(a_s,b_s)|s\in\mathcal{S}\}$ in order to ensure that Theseus achieves good performance. To simplify our analysis, we assume that each worker $s$'s cost function $C_s(\cdot)$ is linear in $\delta_s\in[\underline{\delta}_s, \overline{\delta}_s]$, i.e., 
\begin{equation}
\begin{aligned}
C_s(\delta_s)=-c_{s,1}\delta_s+c_{s,2},~\forall \delta_s\in[\underline{\delta}_s, \overline{\delta}_s],
\end{aligned}
\end{equation}
where $c_{s,1}$ and $c_{s,2}$ are positive parameters. Note that such selection of each worker $s$'s cost function conforms to the requirement that her cost should decrease with the increase of $\delta_s$. 

According to how much prior knowledge the platform has about workers' cost functions, we parameterize Theseus in the following two scenarios, namely the \textit{complete information scenario} where the platform knows exactly each worker $s$'s $c_{s,1}$ and $c_{s,2}$ (Section \ref{sec:complete}), as well as the \textit{incomplete information scenario} where only limited information about $c_{s,1}$ and $c_{s,2}$ is known by the platform (Section \ref{sec:incomplete}). In both scenarios, we assume that $\underline{\delta}_1,\underline{\delta}_2,\cdots,\underline{\delta}_S$, i.e, the lower bounds of workers' $\delta_s$'s, are i.i.d. random variables within the range $[\underline{\delta}, \overline{\delta}]$ with PDF $f(\cdot)$. Furthermore, the PDF $f(\cdot)$ is assumed to be \textit{a priori} known by the platform and workers, which, as will be shown in Section \ref{sec:completeanalysis} and \ref{sec:incompleteanalysis}, is the only prior statistical knowledge needed to evaluate workers' expected utilities. 

\subsection{Complete Information Scenario}\label{sec:complete}
\subsubsection{Parameter Selection}
As aforementioned, in this section, we assume that the platform knows exactly both $c_{s,1}$ and $c_{s,2}$ in each worker $s$'s cost function. Although, in practice, it might be hard for the platform to obtain such exact knowledge, the complete information scenario is still relevant and interesting to study, because it sheds light upon the philosophy of parameterizing Theseus in the incomplete information scenario in Section \ref{sec:incomplete}. For any given $\Delta_t\in[\underline{\delta}, \overline{\delta}]$, we can parameterize Theseus with any set of parameters $\{(a_s,b_s)|s\in\mathcal{S}\}$ that satisfy Condition (\ref{eq:con1})-(\ref{eq:con3}). 
\begin{numcases}{\hspace{-0.5cm}}
\hspace{-0.2cm}a_s\geq\frac{c_{s,1}}{2\underline{\delta}},&$\forall s\in\mathcal{S}$\label{eq:con1}\\
\hspace{-0.2cm}b_s=a_s\big(\Delta_t^2+A(\Delta_t)\big)-c_{s,1}\Delta_t+c_{s,2},&$\forall s\in\mathcal{S}$\label{eq:con2}\\
\hspace{-0.2cm}\sum_{s=1}^S b_s\leq B+\sum_{s=1}^S 2 a_s \underline{\delta}^2,\label{eq:con3}
\end{numcases}
where $A(\Delta_t)=\int_{\underline{\delta}}^{\Delta_t}u^2\frac{f(u)}{\int_{\underline{\delta}}^{\Delta_t}f(v)dv}du$. The criterion of selecting the additional parameter $\Delta_t$ will be discussed in Section \ref{sec:completeanalysis} as we analyze the performance guarantees of the parameter selection given by Condition (\ref{eq:con1})-(\ref{eq:con3}). For each $s\in\mathcal{S}$, as $b_s$ is exactly determined by $a_s$ due to Condition (\ref{eq:con2}), one way of  parameter selection is to choose an $a_s\geq\frac{c_{s,1}}{2\underline{\delta}}$ such that Condition (\ref{eq:con3}) is satisfied. 
\subsubsection{Analysis}\label{sec:completeanalysis} 
In this section, we carry out analyses about the desirable properties of Theseus by parameterizing it according to Condition (\ref{eq:con1})-(\ref{eq:con3}). We derive the BNE of the sensing game that corresponds to such parameterization in the following Theorem \ref{theo:bne}.
\begin{myTheo}\label{theo:bne}
If parameters $\{(a_s,b_s)|s\in\mathcal{S}\}$ satisfy Condition (\ref{eq:con1}) and (\ref{eq:con2}) , we have that $\boldsymbol{\updelta}^*=(\delta_1^*, \delta_2^*,\cdots, \delta_S^*)$, where, for each worker $s\in\mathcal{S}$,
\begin{equation}
\delta_s^*=
\begin{cases}
\begin{aligned}
&\bot,&&\textnormal{if~}\underline{\delta}_s>\Delta_t\\
&\underline{\delta}_s,&&\textnormal{if~}\underline{\delta}_s\leq\Delta_t
\end{aligned},
\end{cases}
\end{equation}
is a BNE of the sensing game in the complete information scenario. 
\end{myTheo}
\begin{proof}
If any worker $s$ chooses to participate, her expected utility, when other workers take strategies $\boldsymbol{\updelta}^*_{-s}$, and her reference worker $r$'s strategy $\delta_r^*$ is given, can be calculated as 
\begin{equation*} 
\begin{aligned}
\mathbb{E}\big[u_s(\delta_s,\boldsymbol{\updelta}_{-s}^*)\big|\delta_r^*\big]=&\mathbb{E}\big[p_s\big(\mathbf{X}(\delta_s,\boldsymbol{\updelta}_{-s}^*)\big)\big|\delta_r^*\big]-C_s(\delta_s)\\
=&b_s-a_s\mathbb{E}\Bigg[\frac{1}{M}\sum_{m=1}^M\big(X_m^s-X_m^r\big)^2\bigg|\delta_r^*\Bigg]\\
&+c_{s,1}\delta_s-c_{s,2}.\\
\end{aligned}
\end{equation*} 
As $X_m^s-X_m^r=\big(X^{\text{truth}}+N(0,\delta_s^2)\big)-\big(X^{\text{truth}}+N(0,\delta_r^2)\big)=N(0,\delta_s^2)-N(0,\delta_r^2)$,  we have $\mathbb{E}\big[(X_m^s-X_m^r)^2\big]=\delta_s^2+\delta_r^2$. Therefore, we have 
\begin{equation*} 
\begin{aligned}
\mathbb{E}\big[u_s(\delta_s,\boldsymbol{\updelta}_{-s}^*)\big|\delta_r^*\big]=b_s-a_s\big(\delta_s^2+(\delta_r^{*})^2\big)+c_{s,1}\delta_s-c_{s,2}, 
\end{aligned}
\end{equation*}
and thus,
\begin{equation*}
\begin{aligned}
\max\Big\{\frac{c_{s,1}}{2a_s}, \underline{\delta}_s\Big\}=\arg\max_{\delta_s\in[\underline{\delta}_s, \overline{\delta}_s]}\mathbb{E}\big[u_s(\delta_s,\boldsymbol{\updelta}_{-s}^*)\big|\delta_r^*\big].
\end{aligned}
\end{equation*}
That is, regardless of the value of $\delta_r^*$, the strategy $\delta_s\in[\underline{\delta}_s, \overline{\delta}_s]$ that maximizes $\mathbb{E}\big[u_s(\delta_s,\boldsymbol{\updelta}_{-s}^*)\big|\delta_r^*\big]$ is the maximum between $\frac{c_{s,1}}{2a_s}$ and $\underline{\delta}_s$. Because of Condition (\ref{eq:con1}), we have that $\underline{\delta}_s\geq\underline{\delta}\geq\frac{c_{s,1}}{2a_s}$. Therefore, if any worker $s$ chooses to participate, her strategy must be $\underline{\delta}_s$, and thus, her expected utility is
\begin{equation*} 
\scalefont{0.9}
\begin{aligned}
\mathbb{E}\big[u_s(\underline{\delta}_s,\boldsymbol{\updelta}_{-s}^*)\big]=&\mathbb{E}_{\delta_r^*}\Big[\mathbb{E}\big[u_s(\underline{\delta}_s,\boldsymbol{\updelta}_{-s}^*)\big|\delta_r^*\big]\Big]=\mathbb{E}_{\underline{\delta}_r}\Big[\mathbb{E}\big[u_s(\underline{\delta}_s,\boldsymbol{\updelta}_{-s}^*)\big|\underline{\delta}_r\big]\Big]\\
=&\mathbb{E}_{\underline{\delta}_r}\big[b_s-a_s\big(\underline{\delta}_s^2+\underline{\delta}_r^2\big)+c_{s,1}\underline{\delta}_s-c_{s,2}\big]\\
=&b_s-a_s\Big(\underline{\delta}_s^2+\mathbb{E}_{\underline{\delta}_r}\big[\underline{\delta}_r^2\big|\underline{\delta}_r\leq\Delta_t\big]\Big)+c_{s,1}\underline{\delta}_s-c_{s,2}\\
=&\Big(a_s\big(\Delta_t^2+A(\Delta_t)\big)-c_{s,1}\Delta_t+c_{s,2}\Big)-\\
&\Big(a_s\big(\underline{\delta}_s^2+A(\Delta_t)\big)-c_{s,1}\underline{\delta}_s+c_{s,2}\Big), 
\end{aligned}
\end{equation*} 
where $A(\Delta_t)=\int_{\underline{\delta}}^{\Delta_t}u^2\frac{f(u)}{\int_{\underline{\delta}}^{\Delta_t}f(v)dv}du$, and the last equality is due to Condition (\ref{eq:con2}). Therefore, we have that for each worker $s\in\mathcal{S}$
\begin{equation*}
\begin{cases}
\begin{aligned}
&\mathbb{E}\big[u_s(\underline{\delta}_s,\boldsymbol{\updelta}_{-s}^*)\big]
<0,&&\textnormal{if~}\underline{\delta}_s>\Delta_t\\
&\mathbb{E}\big[u_s(\underline{\delta}_s,\boldsymbol{\updelta}_{-s}^*)\big]
\geq 0,&&\textnormal{if~}\underline{\delta}_s\leq\Delta_t
\end{aligned},
\end{cases}
\end{equation*}
and thus, given that other workers take the strategies $\boldsymbol{\updelta}_{-s}^*$, worker $s$ will drop out, if $\underline{\delta}_s>\Delta_t$, and will take strategy $\underline{\delta}_s$, if $\underline{\delta}_s\leq\Delta_t$. Hence, the strategy profile $\boldsymbol{\updelta}^*$ given in Theorem \ref{theo:bne} is a BNE of the sensing game. 
\end{proof}

Theorem \ref{theo:bne} gives us a BNE of the sensing game, where every worker $s$ with $\underline{\delta}_s>\Delta_t$ will voluntarily drop out, and as long as $\underline{\delta}_s\leq\Delta_t$, the worker $s$ will participate with strategy $\underline{\delta}_s$, which is exactly the smallest standard deviation of the difference between her data and the ground truths. That is, by satisfying Condition (\ref{eq:con1}) and (\ref{eq:con2}), Theseus will only incentivize workers who potentially is capable of providing high quality data to participate, and those who choose to participate will exert their maximum amount of effort, leading them to provide reliable data. Note that there might be multiple BNEs for the sensing game. However, to the best of our knowledge, we have not found other BNEs except for the one given in Theorem \ref{theo:bne}, on which our further analyses in this section are based. We leave the derivation of other BNEs or the proof of the uniqueness of BNE in our future work. Next, we prove in the following Theorem \ref{theo:budget} that Theseus satisfies budget feasibility in the complete information scenario by satisfying Condition (\ref{eq:con3}). 

\begin{myTheo}\label{theo:budget}
Condition (\ref{eq:con3}) ensures that Theseus is budget feasible in the complete information scenario. 
\end{myTheo}
\begin{proof}
At the BNE $\boldsymbol{\updelta}^*$ given in Theorem \ref{theo:bne}, the expected overall payment satisfies that 
\begin{equation*}
\scalefont{0.82}
\begin{aligned}
\mathbb{E}\Bigg[\sum_{s\in\mathcal{S}'}p_s\big(\mathbf{X}(\boldsymbol{\updelta}^*)\big)\Bigg]&=\sum_{s\in\mathcal{S}'}\mathbb{E}\big[p_s\big(\mathbf{X}(\boldsymbol{\updelta}^*)\big)\big]\\
&=\sum_{s\in\mathcal{S}'}\Big(-a_s\Big(\underline{\delta}_s^2+\mathbb{E}\big[\underline{\delta}_{r_s}^2\big|\underline{\delta}_{r_s}\leq\Delta_t\big]\Big)+b_s\Big)\\
&\leq\sum_{s=1}^S\big(-2a_s\underline{\delta}^2+b_s\big)\leq\sum_{s=1}^S\big(-2a_s\underline{\delta}^2+2a_s\underline{\delta}^2\big)+B=B,
\end{aligned}
\end{equation*}
where the last inequality is because of Condition (\ref{eq:con3}), which exactly proves that Theseus is budget feasible in the complete information scenario. 
\end{proof}

Clearly, as stated in the following Theorem \ref{theo:ir}, Theseus satisfies individual rationality in the complete information scenario. 

\begin{myTheo}\label{theo:ir}
Theseus is individual rational in the complete information scenario. 
\end{myTheo}
\begin{proof}
Theorem \ref{theo:ir} is an obvious fact, which directly follows from the fact that only workers with non-negative expected utilities at the BNE will choose to participate. Hence, no worker will have negative utility, and thus, individual rationality is satisfied in the complete information scenario. 
\end{proof}

Next, we discuss our selection criterion of the parameter $\Delta_t$. Following notational conventions in order statistics, we denote $\underline{\delta}_{(1)}=\min\{\underline{\delta}_1,\underline{\delta}_2,\cdots,\underline{\delta}_S\}$. We assume that the CDF $F(\cdot)$ of any $\underline{\delta}_s$ is invertible, and its inverse is $F^{-1}(\cdot)$. Based on Theorem \ref{theo:bne}, if $\Delta_t$ is set to be too small, no workers will participate at the BNE. Thus, we establish a lower bound for $\Delta_t$ in the following Theorem \ref{theo:completelb}. 

\begin{myTheo}\label{theo:completelb}
Given any $\theta_c\in(0,1)$, if $\Delta_t\geq F^{-1}\big(1-\sqrt[S]{1-\theta_c}\big)$, then $\textnormal{Pr}\big(\underline{\delta}_{(1)}\leq\Delta_t\big)\geq\theta_c$, i.e., the probability that at least one worker chooses to participate at the BNE of the sensing game, in the complete information scenario, is no less than the threshold $\theta_c$. 
\end{myTheo}

\begin{proof}
For any given $\Delta_t$, we have that 
\begin{align*}
\textnormal{Pr}\big(\underline{\delta}_{(1)}\leq\Delta_t\big)&=\textnormal{Pr}\big(\min\{\underline{\delta}_1,\underline{\delta}_2,\cdots,\underline{\delta}_S\}\leq\Delta_t\big)\\
&=1-\textnormal{Pr}\big(\min\{\underline{\delta}_1,\underline{\delta}_2,\cdots,\underline{\delta}_S\}>\Delta_t\big)\\
&=1-\prod_{s=1}^S\text{Pr}\big(\underline{\delta}_s>\Delta_t\big)=1-\big(1-F(\Delta_t)\big)^S. 
\end{align*}
Thus, for any $\theta_c\in(0,1)$, we get $\Delta_t\geq F^{-1}\big(1-\sqrt[S]{1-\theta_c}\big)$ by setting $1-\big(1-F(\Delta_t)\big)^S\geq\theta_c$, which proves Theorem \ref{theo:completelb}. 
\end{proof}

In the rest of our analyses, we use APP to denote the the value of the PMD problem's objective function guaranteed by Theseus. Theorem \ref{theo:completelb} gives us that if $\Delta_t\geq F^{-1}\big(1-\sqrt[S]{1-\theta_c}\big)$, the probability that there exists at least one participating worker at the BNE of the sensing game is guaranteed to be no less than the predefined threshold $\theta_c\in(0,1)$. However, this does not mean that $\Delta_t$ could be infinitely large, because the greater $\Delta_t$ is, the farther APP will drift apart from the minimum value of the PMD problem's objective function. Thus, in the following Theorem \ref{theo:completeub}, we derive an upper bound for the parameter $\Delta_t$. Note that for any payment mechanism that ensures the participation of at least one worker, the minimum possible value for the objective function is $\text{OPT}=\underline{\delta}_{(1)}$, which is the optimal benchmark that we compare APP with. 

\begin{myTheo}\label{theo:completeub}
In the complete information scenario, given $\alpha_c>1$ and $\beta_c\in(0,1)$, we have that 
\begin{equation}
\begin{aligned}
\textnormal{Pr}\bigg(\frac{\textnormal{APP}}{\textnormal{OPT}}\geq\alpha_c\bigg)\leq\beta_c, 
\end{aligned}
\end{equation}
if $\Delta_t\leq\overline{\Delta}_t$, where $\overline{\Delta}_t$ is the solution to 
\begin{equation}
\begin{aligned}
\overline{\Delta}_t+\sqrt{-\frac{2}{S\ln\beta_c}}\Big(R\big(\overline{\Delta}_t\big)S-\underline{\delta}\alpha_c\Big)=0, 
\end{aligned}
\end{equation}
with $R\big(\overline{\Delta}_t\big)=\int_{\underline{\delta}}^{\overline{\Delta}_t}u\frac{f(u)}{\int_{\underline{\delta}}^{\overline{\Delta}_t}f(v)dv}du$. 
\end{myTheo}
\begin{proof}
At the BNE $\boldsymbol{\updelta}^*$ given in Theorem \ref{theo:bne}, we have that
\begin{equation*}
\begin{aligned}
\text{APP}=\sum_{s\in\mathcal{S}'}\delta_s^*=\sum_{s\in\mathcal{S}'}\underline{\delta}_s=\sum_{s=1}^S \underline{\delta}_s\mathds{1}_{\{\underline{\delta}_s\leq\Delta_t\}},
\end{aligned}
\end{equation*}
where, for each $s\in\mathcal{S}$, $\mathds{1}_{\{\underline{\delta}_s\leq\Delta_t\}}$ is an indicator function, such that 
\begin{equation*}
\mathds{1}_{\{\underline{\delta}_s\leq\Delta_t\}}=
\begin{cases}
\begin{aligned}
&0,&&\text{if~}\underline{\delta}_s>\Delta_t\\
&1,&&\textnormal{if~}\underline{\delta}_s\leq\Delta_t
\end{aligned},
\end{cases}
\end{equation*}
and thus, $\underline{\delta}_s\mathds{1}_{\{\underline{\delta}_s\leq\Delta_t\}}\in[0,\Delta_t]$. Thus, for a fixed $\alpha_c>1$, we have 
\begin{equation*}
\scalefont{0.9}
\begin{aligned}
&\text{Pr}\bigg(\frac{\textnormal{APP}}{\textnormal{OPT}}\geq\alpha_c\bigg)\\
=&\text{Pr}\Bigg(\frac{\sum_{s=1}^S \underline{\delta}_s\mathds{1}_{\{\underline{\delta}_s\leq\Delta_t\}}}{\underline{\delta}_{(1)}}\geq\alpha_c\Bigg)\leq\text{Pr}\Bigg(\frac{\sum_{s=1}^S \underline{\delta}_s\mathds{1}_{\{\underline{\delta}_s\leq\Delta_t\}}}{\underline{\delta}}\geq\alpha_c\Bigg)\\
=&\text{Pr}\Bigg(\frac{\sum_{s=1}^S\big(\underline{\delta}_s\mathds{1}_{\{\underline{\delta}_s\leq\Delta_t\}}-\mathbb{E}[\underline{\delta}_s|\underline{\delta}_s\leq\Delta_t]\big)}{S}\geq\frac{\alpha_c\underline{\delta}}{S}-\mathbb{E}\big[\underline{\delta}_1\big|\underline{\delta}_1\leq\Delta_t\big]\Bigg)\\
\leq&\exp\Bigg(-\frac{2 S^2\big(\frac{\alpha_c\underline{\delta}}{S}-\mathbb{E}[\underline{\delta}_1|\underline{\delta}_1\leq\Delta_t]\big)^2}{S\Delta_t^2}\Bigg)=\exp\Bigg(-\frac{2(\alpha_c\underline{\delta}-R(\Delta_t)S)^2}{S\Delta_t^2}\Bigg), 
\end{aligned}
\end{equation*}
where the last inequality is because of the Hoeffding's inequality, and $R(\Delta_t)=\int_{\underline{\delta}}^{\Delta_t}u\frac{f(u)}{\int_{\underline{\delta}}^{\Delta_t}f(v)dv}du$. For a fixed $\beta_c\in(0,1)$, by setting $\exp\Big(-\frac{2(\alpha_c\underline{\delta}-R(\Delta_t)S)^2}{S\Delta_t^2}\Big)\leq\beta_c$, we get that $\Delta_t+\sqrt{-\frac{2}{S\ln\beta_c}}\big(R(\Delta_t)S-\underline{\delta}\alpha_c\big)\leq 0$. Therefore, by setting $\Delta_t$ to be no greater than the upper bound $\overline{\Delta}_t$ given in Theorem \ref{theo:completeub}, we have $\textnormal{Pr}\big(\frac{\textnormal{APP}}{\textnormal{OPT}}\geq\alpha_c\big)\leq\beta_c$. 
\end{proof}

By Theorem \ref{theo:completeub}, we have that, as long as $\Delta_t\leq\overline{\Delta}_t$, the probability that the approximation ratio $\frac{\text{APP}}{\text{OPT}}\geq\alpha_c$ is no greater than $\beta_c$, for the predefined constants $\alpha_c>1$ and $\beta_c\in(0,1)$. This shows the probabilistic guarantee on the approximation ratio of Theseus compared to the optimal payment mechanism. Next, we have the following Corollary \ref{theo:cor1} about the range from which the parameter $\Delta_t$ should be selected.

\begin{myCor}\label{theo:cor1}
By jointly considering Theorem \ref{theo:completelb} and \ref{theo:completeub}, the parameter $\Delta_t$ should satisfy that $F^{-1}\big(1-\sqrt[S]{1-\theta_c}\big)\leq\Delta_t\leq\overline{\Delta}_t$ in the complete information scenario, in order to guarantee that with high probability there exist participating workers at the corresponding BNE (Theorem \ref{theo:completelb}), and that with high probability Theseus has a small approximation ratio (Theorem \ref{theo:completeub}).
\end{myCor}

\subsection{Incomplete Information Scenario}\label{sec:incomplete}
\subsubsection{Parameter Selection}
In this section, we study a more practical incomplete information scenario, where the platform does not know the exact values of each worker $s$'s $c_{s,1}$ and $c_{s,2}$, but instead, only knows that $c_{s,1}\in[\underline{c}_1, \overline{c}_1]$, and $c_{s,2}\in[\underline{c}_2, \overline{c}_2]$, for each worker $s$. In this case, given any $\Delta_l$ and $\Delta_h$, such that $\underline{\delta}\leq\Delta_l<\Delta_h\leq\overline{\delta}$, we can parameterize Theseus with any set of parameters $\{(a_s,b_s)|s\in\mathcal{S}\}$ such that Condition (\ref{eq:con4})-(\ref{eq:con7}) are satisfied. 
\begin{numcases}{\hspace{-0.5cm}}
\hspace{-0.2cm}a_s\geq \frac{\overline{c}_1}{2\underline{\delta}},&$\forall s\in\mathcal{S}$\label{eq:con4}\\
\hspace{-0.2cm}b_s\leq a_s\big(\Delta_h^2+A(\Delta_h)\big)-\overline{c}_1\Delta_h+\underline{c}_2,&$\forall s\in\mathcal{S}$\label{eq:con5}\\
\hspace{-0.2cm}b_s\geq a_s\big(\Delta_l^2+A(\Delta_h)\big)-\underline{c}_1\Delta_l+\overline{c}_2,,&$\forall s\in\mathcal{S}$\label{eq:con6}\\
\hspace{-0.2cm}\sum_{s=1}^S b_s\leq B+\sum_{s=1}^S 2 a_s \underline{\delta}^2,\label{eq:con7}
\end{numcases}
where $A(\Delta_h)=\int_{\underline{\delta}}^{\Delta_h}u^2\frac{f(u)}{\int_{\underline{\delta}}^{\Delta_h}f(v)dv}du$. Note that the criterion of selecting $\Delta_l$ and $\Delta_h$ will be discussed in Section \ref{sec:incompleteanalysis} as we introduce the corresponding analyses. Given these conditions, one specific way of parameter selection for each $s\in\mathcal{S}$ is to choose an $a_s\geq \frac{\overline{c}_1}{2\underline{\delta}}$ such that Condition (\ref{eq:con4})-(\ref{eq:con7}) are satisfied. 

\subsubsection{Analysis}\label{sec:incompleteanalysis}
In this section, we firstly characterize the BNE of the sensing game by parameterizing Theseus in the incomplete information scenario according to Condition (\ref{eq:con4})-(\ref{eq:con7}) in the following Theorem \ref{theo:bnein}.

\begin{myTheo}\label{theo:bnein}
If parameters $\{(a_s,b_s)|s\in\mathcal{S}\}$ satisfy Condition (\ref{eq:con4})-(\ref{eq:con6}) , we have a BNE $\boldsymbol{\updelta}^*=(\delta_1^*, \delta_2^*,\cdots, \delta_S^*)$ of the sensing game in the incomplete information scenario, such that, for each worker $s\in\mathcal{S}$, 
\begin{equation}\label{eq:bneincomplete}
\delta_s^*=
\begin{cases}
\begin{aligned}
&\bot,&&\textnormal{if~}\underline{\delta}_s>\Delta_h\\
&\underline{\delta}_s,&&\textnormal{if~}\underline{\delta}_s\leq\Delta_l
\end{aligned}.
\end{cases}
\end{equation}
\end{myTheo}

\begin{proof}
From Condition (\ref{eq:con4}), we have that $\underline{\delta}_s\geq\underline{\delta}\geq\frac{\overline{c}_1}{2a_s}\geq\frac{c_{s,1}}{2 a_s}$ for each worker $s$. Thus, based on exactly the same reasoning as in the proof of Theorem \ref{theo:bne}, if a worker $s$ chooses to participate when other workers take strategies $\boldsymbol{\updelta}_{-s}^*$, her strategy must be $\underline{\delta}_s$. 

Therefore, given that other workers take strategies $\boldsymbol{\updelta}_{-s}^*$, worker $s$'s expected utility is 
\begin{equation*} 
\begin{aligned}
\mathbb{E}\big[u_s(\underline{\delta}_s,\boldsymbol{\updelta}_{-s}^*)\big]&=b_s-a_s\Big(\underline{\delta}_s^2+\mathbb{E}\big[\underline{\delta}_r^2\big|\underline{\delta}_r\leq\Delta_h\big]\Big)+c_{s,1}\underline{\delta}_s-c_{s,2}\\
&=b_s-a_s\big(\underline{\delta}_s^2+A(\Delta_h)\big)+c_{s,1}\underline{\delta}_s-c_{s,2}.
\end{aligned}
\end{equation*} 

Thus, by Condition (\ref{eq:con5}), for any worker $s$ with $\underline{\delta}_s>\Delta_h$, we have that 
\begin{equation*} 
\begin{aligned}
\mathbb{E}\big[u_s(\underline{\delta}_s,\boldsymbol{\updelta}_{-s}^*)\big]\leq&\big(a_s\big(\Delta_h^2+A(\Delta_h)\big)-\overline{c}_1\Delta_h+\underline{c}_2\big)-\\
&\big(a_s\big(\underline{\delta}_s^2+A(\Delta_h)\big)-c_{s,1}\underline{\delta}_s+c_{s,2}\big)<0
\end{aligned}
\end{equation*} 
and by Condition (\ref{eq:con6}), for any worker $s$ with $\underline{\delta}_s\leq\Delta_l$, we have that 
\begin{equation*} 
\begin{aligned}
\mathbb{E}\big[u_s(\underline{\delta}_s,\boldsymbol{\updelta}_{-s}^*)\big]\geq&\big(a_s\big(\Delta_l^2+A(\Delta_h)\big)-\underline{c}_1\Delta_l+\overline{c}_2\big)-\\
&\big(a_s\big(\underline{\delta}_s^2+A(\Delta_h)\big)-c_{s,1}\underline{\delta}_s+c_{s,2}\big)\geq0.
\end{aligned}
\end{equation*}

Thus, given that other workers take strategies $\boldsymbol{\updelta}_{-s}^*$, worker $s$ will drop out, if $\underline{\delta}_s>\Delta_h$, and will take strategy $\underline{\delta}_s$, if $\underline{\delta}_s\leq\Delta_l$. Hence, there exists a BNE $\boldsymbol{\updelta}^*$ of the sensing game such that Equation (\ref{eq:bneincomplete}) is satisfied for each worker $s$. 
\end{proof}

Theorem \ref{theo:bnein} characterizes a BNE of the sensing game, where each worker $s$ with $\underline{\delta}_s>\Delta_h$ will drop out, and as long as $\underline{\delta}_s\leq\Delta_l$, she will participate with strategy $\underline{\delta}_s$. Note that, at the BNE, each worker $s$ with $\underline{\delta}_s\in(\Delta_l,\Delta_h]$ has to evaluate her expected utility based on the specific choice of $\{(a_s,b_s)|s\in\mathcal{S}\}$ in order to make the decision of whether or not to participate. All of the following analyses in this section are based on the BNE characterized in Theorem \ref{theo:bnein}. We also leave the proof of the uniqueness of BNE or the derivation of other BNEs in our future work. Next, we introduce in Theorem \ref{theo:budgetincomplete} and \ref{theo:irincomplete} about the budget feasibility and individual rationality of Theseus in the incomplete information scenario.

\begin{myTheo}\label{theo:budgetincomplete}
Condition (\ref{eq:con7}) ensures that Theseus is budget feasible in the incomplete information scenario. 
\end{myTheo}

\begin{myTheo}\label{theo:irincomplete}
Theseus is individual rational in the incomplete information scenario. 
\end{myTheo}

The proof of Theorem \ref{theo:budgetincomplete} is the same as that of Theorem \ref{theo:budget} except that the $\Delta_t$ is replaced by $\Delta_h$, and the proof of Theorem \ref{theo:irincomplete} is exactly identical to that of Theorem \ref{theo:ir}. Thus, we omit the formal proofs of Theorem \ref{theo:budgetincomplete} and \ref{theo:irincomplete} in this paper. Similar to Section \ref{sec:completeanalysis}, we establish ranges from which we select parameters $\Delta_l$ and $\Delta_h$. In the following Theorem \ref{theo:incompletelb}, we introduce a lower bound for $\Delta_l$. 

\begin{myTheo}\label{theo:incompletelb}
Given any $\theta_{ic}\in(0,1)$, if $\Delta_l\geq F^{-1}\big(1-\sqrt[S]{1-\theta_{ic}}\big)$, then $\textnormal{Pr}\big(\underline{\delta}_{(1)}\leq\Delta_l\big)\geq\theta_{ic}$, i.e., the probability that at least one worker chooses to participate at the BNE of the sensing game, in the incomplete information scenario, is no less than the threshold $\theta_{ic}$. 
\end{myTheo}

The proof of Theorem \ref{theo:incompletelb} is omitted in this paper as well, because it can be directly adapted from that of Theorem \ref{theo:completelb} by changing $\Delta_t$ to $\Delta_l$ and $\theta_c$ to $\theta_{ic}$. By Theorem \ref{theo:incompletelb}, we have that, in the incomplete information scenario, it is $\Delta_l$ that decides the probability that at least one worker chooses to participate at the BNE of the sensing game. That is, as long as $\Delta_l\geq F^{-1}\big(1-\sqrt[S]{1-\theta_{ic}}\big)$, this probability, i.e., $\textnormal{Pr}\big(\underline{\delta}_{(1)}\leq\Delta_l\big)$, will be no less than the predefined threshold $\theta_{ic}$. Next, in the following Theorem \ref{theo:incompleteub}, where APP and OPT have the same meanings as in Theorem \ref{theo:completeub}, we derive an upper bound for the parameter $\Delta_h$.  

\begin{myTheo}\label{theo:incompleteub}
In the incomplete information scenario, given $\alpha_{ic}>1$ and $\beta_{ic}\in(0,1)$, we have that 
\begin{equation}
\begin{aligned}
\textnormal{Pr}\bigg(\frac{\textnormal{APP}}{\textnormal{OPT}}\geq\alpha_{ic}\bigg)\leq\beta_{ic}, 
\end{aligned}
\end{equation}
if $\Delta_h\leq\overline{\Delta}_h$, where $\overline{\Delta}_h$ is the solution to 
\begin{equation}
\begin{aligned}
\overline{\Delta}_h+\sqrt{-\frac{2}{S\ln\beta_{ic}}}\Big(R\big(\overline{\Delta}_h\big)S-\underline{\delta}\alpha_{ic}\Big)=0, 
\end{aligned}
\end{equation}
with $R\big(\overline{\Delta}_h\big)=\int_{\underline{\delta}}^{\overline{\Delta}_h}u\frac{f(u)}{\int_{\underline{\delta}}^{\overline{\Delta}_h}f(v)dv}du$. 
\end{myTheo}

\begin{proof}
At the BNE $\boldsymbol{\updelta}^*$ of the sensing game characterized in Theorem \ref{theo:bnein}, we have that 
\begin{equation*}
\begin{aligned}
\text{APP}=\sum_{s\in\mathcal{S}'}\underline{\delta}_s\leq\sum_{s=1}^S \underline{\delta}_s\mathds{1}_{\{\underline{\delta}_s\leq\Delta_h\}},
\end{aligned}
\end{equation*}
where, for each $s\in\mathcal{S}$, $\mathds{1}_{\{\underline{\delta}_s\leq\Delta_h\}}$ is an indicator function, such that 
\begin{equation*}
\mathds{1}_{\{\underline{\delta}_s\leq\Delta_h\}}=
\begin{cases}
\begin{aligned}
&0,&&\text{if~}\underline{\delta}_s>\Delta_h\\
&1,&&\textnormal{if~}\underline{\delta}_s\leq\Delta_h
\end{aligned}.
\end{cases}
\end{equation*}

Thus, similar to the proof of Theorem \ref{theo:completeub}, for a fixed $\alpha_{ic}>1$, we have that 
\begin{equation*}
\begin{aligned}
\text{Pr}\bigg(\frac{\textnormal{APP}}{\textnormal{OPT}}\geq\alpha_{ic}\bigg)&\leq\text{Pr}\Bigg(\frac{\sum_{s=1}^S \underline{\delta}_s\mathds{1}_{\{\underline{\delta}_s\leq\Delta_h\}}}{\underline{\delta}_{(1)}}\geq\alpha_{ic}\Bigg)\\
&\leq\exp\Bigg(-\frac{2(\alpha_{ic}\underline{\delta}-R(\Delta_h)S)^2}{S\Delta_h^2}\Bigg),
\end{aligned}
\end{equation*}
where $R(\Delta_h)=\int_{\underline{\delta}}^{\Delta_h}u\frac{f(u)}{\int_{\underline{\delta}}^{\Delta_h}f(v)dv}du$. Thus, for any fixed $\beta_{ic}\in(0,1)$, by setting $\exp\Big(-\frac{2(\alpha_{ic}\underline{\delta}-R(\Delta_h)S)^2}{S\Delta_h^2}\Big)\leq\beta_{ic}$, we get $\Delta_h+\sqrt{-\frac{2}{S\ln\beta_{ic}}}\big(R(\Delta_h)S-\underline{\delta}\alpha_{ic}\big)\leq 0$. Therefore, by setting $\Delta_h$ to be no greater than the upper bound $\overline{\Delta}_h$ given in Theorem \ref{theo:incompleteub}, we have that $\textnormal{Pr}\big(\frac{\textnormal{APP}}{\textnormal{OPT}}\geq\alpha_{ic}\big)\leq\beta_{ic}$.
\end{proof}

\vspace{-0.02cm}
Similar to Theorem \ref{theo:completeub}, Theorem \ref{theo:incompleteub} gives us a probabilistic guarantee on the approximation ratio of Theseus compared to the optimal payment mechanism in the incomplete information scenario. That is, as long as $\Delta_h\leq\overline{\Delta}_h$, the probability that the approximation ratio $\frac{\textnormal{APP}}{\textnormal{OPT}}\geq\alpha_{ic}$ is no greater than $\beta_{ic}$, for the predefined constants $\alpha_{ic}>1$ and $\beta_{ic}\in(0,1)$. Next, we introduce in Corollary \ref{theo:cor2} about the ranges from which we select the parameters $\Delta_l$ and $\Delta_h$. 

\begin{myCor}\label{theo:cor2}
By jointly considering Theorem \ref{theo:incompletelb} and \ref{theo:incompleteub}, in the incomplete information scenario, the parameters $\Delta_l$ and $\Delta_h$ should satisfy $F^{-1}\big(1-\sqrt[S]{1-\theta_{ic}}\big)\leq\Delta_l<\Delta_h\leq\overline{\Delta}_h$, in order to guarantee, with high probability, the existence of at least one participating worker at the corresponding BNE (Theorem \ref{theo:incompletelb}), and that with high probability Theseus yields a small approximation ratio (Theorem \ref{theo:incompleteub}).
\end{myCor}

\subsection{Summary of Parameterization}
Thus far, we have finished our discussion of parameterizing Theseus in both the complete (Section \ref{sec:complete}) and incomplete (Section \ref{sec:incomplete}) information scenario. In summary, in the complete information scenario, if parameters $\{(a_s,b_s)|s\in\mathcal{S}\}$ and $\Delta_t$ satisfy Condition (\ref{eq:con1})-(\ref{eq:con3}) and Corollary \ref{theo:cor1}, at the BNE derived in Theorem \ref{theo:bne}, Theseus satisfies budget feasibility (Theorem \ref{theo:budget}), individual rationality (Theorem \ref{theo:ir}), as well as with high probability it has a small approximation ratio (Theorem \ref{theo:completeub}), and with high probability it guarantees that there exist participating workers (Theorem \ref{theo:completelb}). Similarly, in the incomplete information scenario, if we set parameters $\{(a_s,b_s)|s\in\mathcal{S}\}$, $\Delta_l$, and $\Delta_h$ according to Condition (\ref{eq:con4})-(\ref{eq:con7}) and Corollary \ref{theo:cor2}, at the BNE characterized in Theorem \ref{theo:bnein}, Theseus also satisfies budget feasibility (Theorem \ref{theo:budgetincomplete}), individual rationality (Theorem \ref{theo:irincomplete}), as well as with high probability it guarantees that there will be participating workers (Theorem \ref{theo:incompletelb}), and with high probability it has a small approximation ratio (Theorem \ref{theo:incompleteub}).

\section{Performance Evaluation} 
In this section, we introduce the baseline methods, as well as simulation settings and results. 
\subsection{Baseline Methods}
In the first baseline method, called \textit{Max Std}, considered in this paper, each worker $s$ takes strategy $\overline{\delta}_s$, i.e., the maximum standard deviation of the difference between her data and the ground truths. Max Std actually corresponds to the family of payment mechanisms that provide rather insufficient incentives so that workers are only willing to spend little amount of effort. Different from Max Std, in the second baseline method, called \textit{Random Std}, each worker $s$ selects her strategy $\delta_s$ uniformly at random from the range $[\underline{\delta}_s,\overline{\delta}_s]$. We compare these two baseline methods with the BNEs of the sensing game induced by Theseus in both the complete and incomplete information scenario, which are established in Theorem \ref{theo:bne} and \ref{theo:bnein}, respectively. Note that we do not compare Theseus with existing mechanisms in past literature, because, as indicated in Section \ref{sec:related}, none of them consider the same scenario as this paper, and thus they are not comparable with Theseus.

\subsection{Simulation Settings}

For the complete information scenario, we consider setting I and II given in Table \ref{tab:setting}. In both of these two settings, for each worker $s$, $\underline{\delta}_s$ is generated uniformly at random from the range $[0.1,4]$, i.e., $\underline{\delta}_s\sim U[0.1,4]$. Furthermore, we set $\theta_c=0.9$, $\alpha_c=5$, and $\beta_c=0.1$, and generate $\overline{\delta}_s$ and $x_m^{\text{truth}}$ uniformly at random from the range $[5,10]$ and $[0,10]$, respectively. In setting I, we fix the number of tasks as $M=30$ and vary the number of workers $S$ from $120$ to $150$, whereas in setting II, we fix the number of workers as $S=130$ and vary the number of tasks $M$ from $10$ to $40$. Note that the parameter $\Delta_t$ is generated uniformly at random from the range $[F^{-1}(1-\sqrt[S]{1-\theta_c}),\overline{\Delta}_t]$. In setting III and IV for the incomplete information scenario, we generate the parameters $\underline{\delta}_s$, $\overline{\delta}_s$, $\theta_{ic}$, $\alpha_{ic}$, $\beta_{ic}$, $x_m^{\text{truth}}$, $S$, and $M$ in the same way as in setting I and II, and select $\Delta_l$ and $\Delta_h$ uniformly at random from the range $[F^{-1}(1-\sqrt[S]{1-\theta_{ic}}),\overline{\Delta}_h]$ such that $\Delta_l<\Delta_h$. 

\begin{table}[h]\footnotesize\setlength{\tabcolsep}{2.1pt}
\centering
\begin{tabular}{|c|c|c|c|c|c|c|c|c|}
\hline
Setting&$\underline{\delta}_s$&$\overline{\delta}_s$&$\theta_c, \theta_{ic}$&$\alpha_c, \alpha_{ic}$&$\beta_c, \beta_{ic}$&$x_m^{\text{truth}}$&$S$&$M$\\
\hline\hline
I, III&$[0.1,4]$&$[5,10]$&$0.9$&$5$&$0.1$&$[0,10]$&$[120,150]$&$30$\\
\hline
II, IV&$[0.1,4]$&$[5,10]$&$0.9$&$5$&$0.1$&$[0,10]$&$130$&$[10, 40]$\\
\hline
\end{tabular}
\caption{Simulation settings}\label{tab:setting}
\end{table}

\vspace{-0.4cm}
In all these settings, given $\delta_s$ and $x_m^{\text{truth}}$, worker $s$'s data on task $m$, which is $x_m^s$, is generated by adding a randomly sampled noise from the distribution $N(0, \delta_s^2)$ to the ground truth $x_m^{\text{truth}}$. Then, workers' data generated by Max Std and Random Std, as well as at the BNEs of the sensing game induced by Theseus, are treated as the inputs to a truth discovery algorithm, respectively, to calculate the estimated ground truths. In our simulation, the truth discovery algorithm that we implement is the widely adopted CRH \cite{QiSIGMOD14}, which calculates each participating worker $s$'s weight $w_s$ as
\begin{equation}
\begin{aligned}
w_s=\log\Bigg(\frac{\sum_{s'\in\mathcal{S}'}\sum_{m\in\mathcal{M}}|x_m^{s'}-x_m^*|^2}{\sum_{m\in\mathcal{M}}|x_m^s-x_m^*|^2}\Bigg).
\end{aligned}
\end{equation}

\subsection{Simulation Results}

In this section, we firstly demonstrate our simulation results regarding the comparison among Max Std, Random Std, and Theseus, in terms of the truth discovery algorithm CRH's MAEs (defined in Definition \ref{def:ep}), in Figure \ref{fig:MAEFixTaskC}-\ref{fig:MAEFixWorkerIC}. Note that for each specific worker and task number, we repeatedly generate workers' data, run the truth discovery algorithm CRH, and calculate the corresponding MAE for 10000 times.  

\begin{figure}[h]
\begin{minipage}{.49\linewidth}
\centering
\includegraphics[width=.98\textwidth]{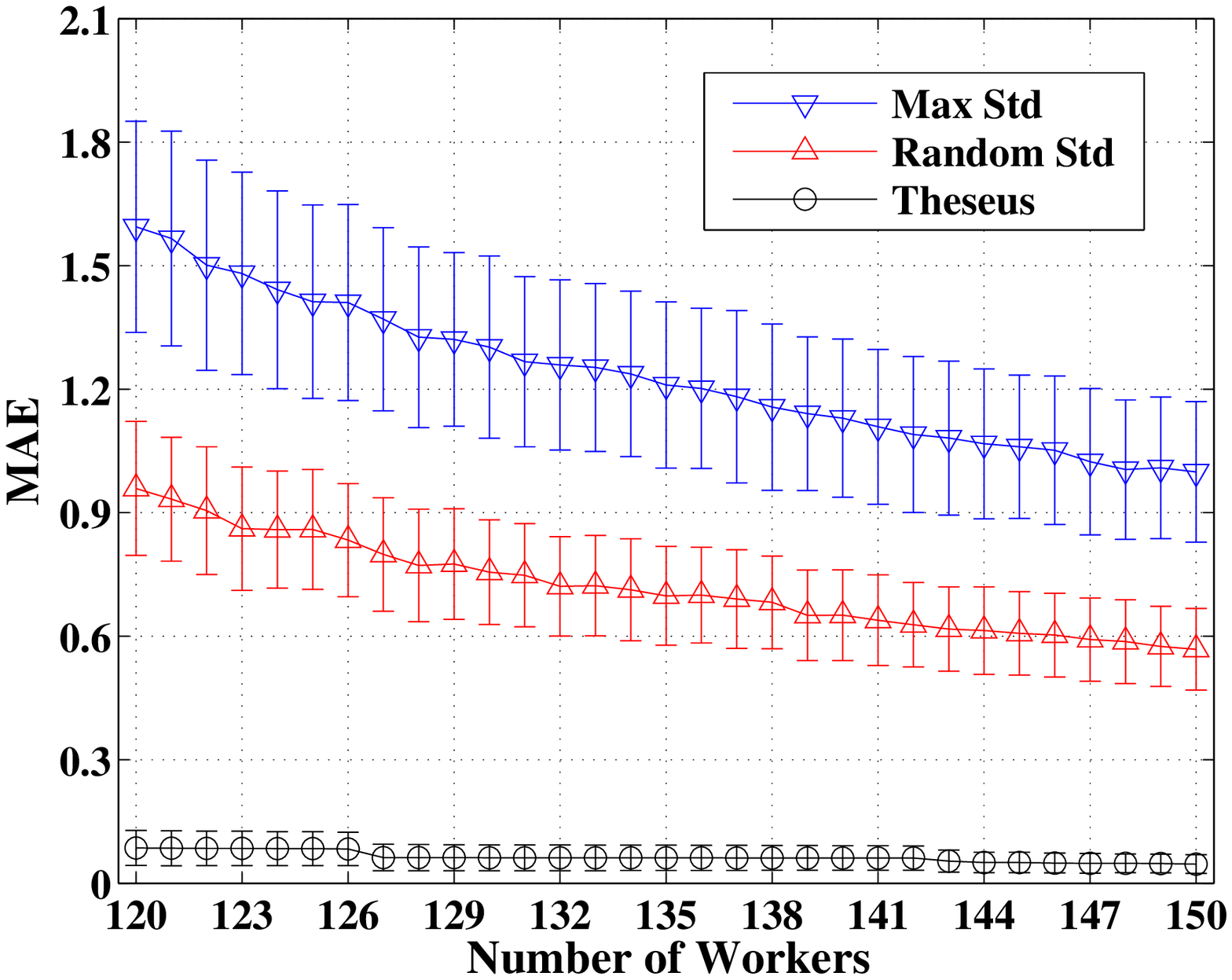}\\
\vspace{-0.45cm}
\caption{MAE comparison (setting I)}
\label{fig:MAEFixTaskC}
\end{minipage}
\begin{minipage}{.49\linewidth}
\centering
\includegraphics[width=.98\textwidth]{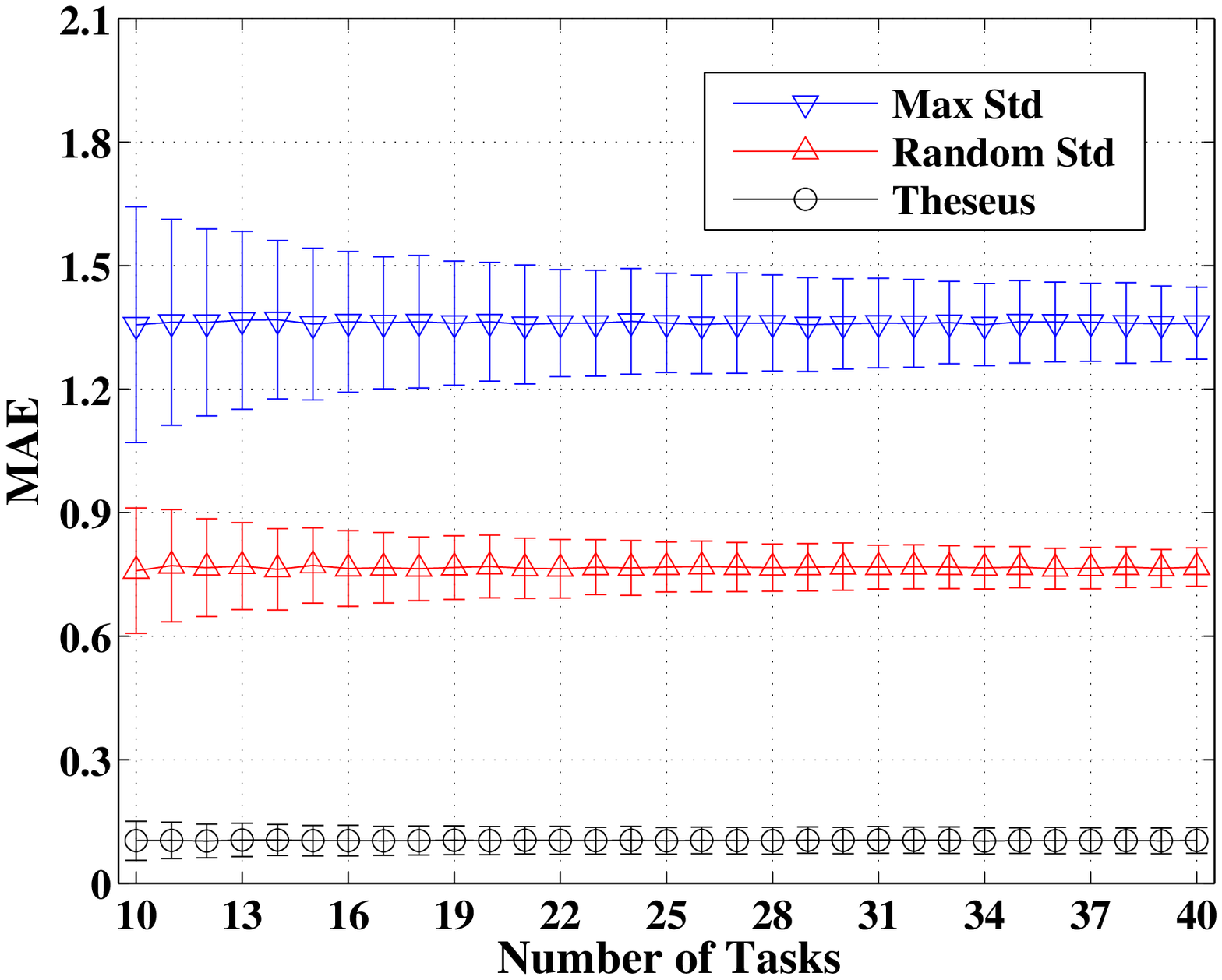}\\
\vspace{-0.45cm}
\caption{MAE comparison (setting II)}
\label{fig:MAEFixWorkerC}
\end{minipage}
\end{figure}

\vspace{-0.05cm}
In Figure \ref{fig:MAEFixTaskC} and \ref{fig:MAEFixWorkerC}, we plot the means and standard deviations of the MAEs corresponding to Max Std, Random Std, and Theseus for setting I and II of the complete information scenario. From these two figures, it is easily observable that the means and standard deviations of the MAEs that correspond to Theseus are far less than those that correspond to Max Std and Random Std, which is because Theseus incentivizes workers to exert their maximum amount of effort, so that the standard deviation between each worker's data and the ground truths is minimized. Note that, in Figure \ref{fig:MAEFixTaskC}, the mean of MAE largely decreases as the number of workers increase, because more data that are close to the ground truths will be inputted to CRH with more number of workers. Figure \ref{fig:MAEFixTaskIC} and \ref{fig:MAEFixWorkerIC} demonstrate similar trends for the MAEs in setting III and IV of the incomplete information scenario. 

Basically, Figure \ref{fig:MAEFixTaskC}-\ref{fig:MAEFixWorkerIC} indicate collectively that a truth discovery algorithm will return rather inaccurate aggregated results, when a vast majority of the participating workers provide unreliable data. Therefore, our Theseus payment mechanism is highly necessary in order to achieve high aggregation accuracy, even though the platform aggregates workers' data using a state-of-the-art truth discovery algorithm. 

\begin{figure}[h]
\begin{minipage}{.49\linewidth}
\centering
\includegraphics[width=.98\textwidth]{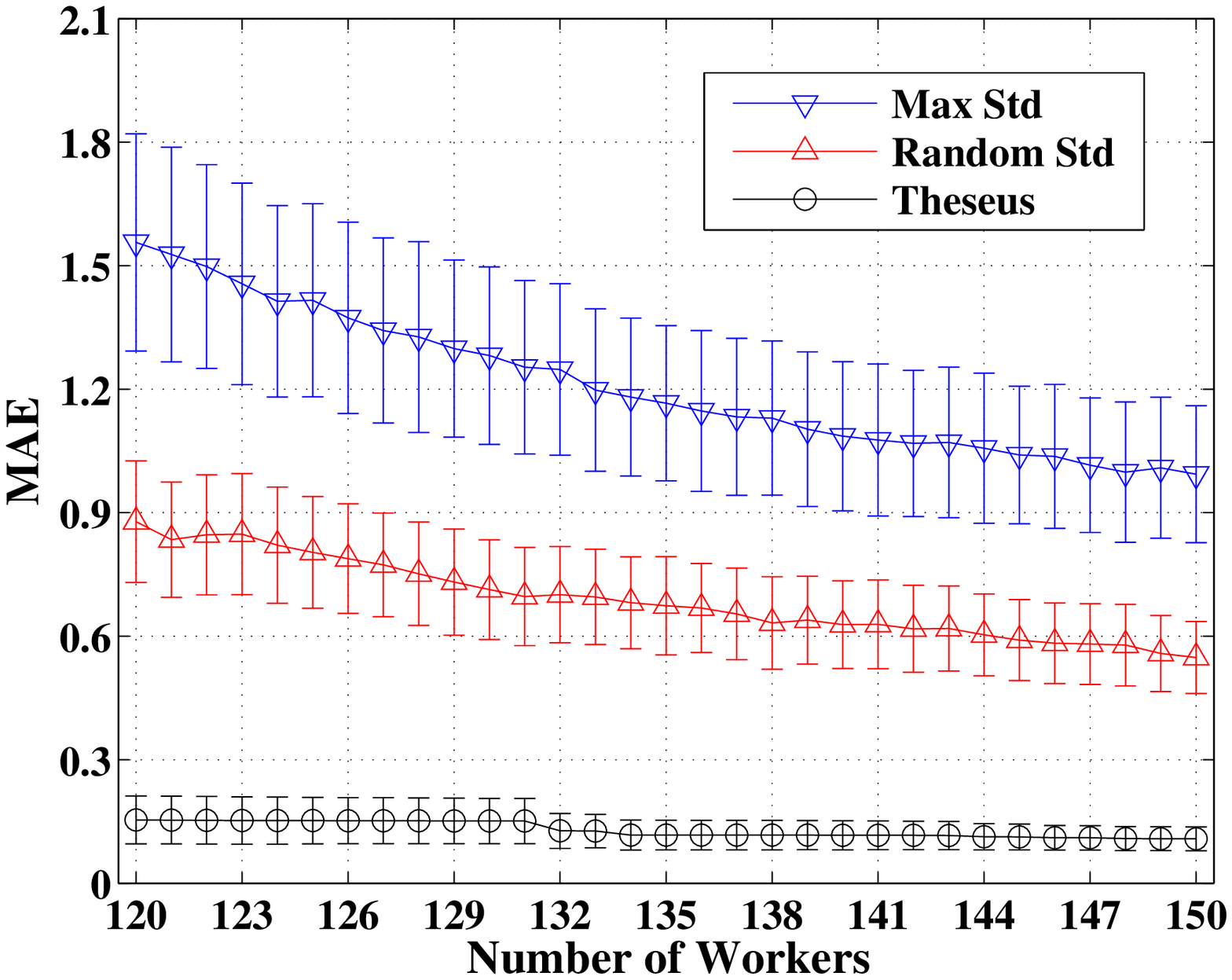}\\
\vspace{-0.45cm}
\caption{MAE comparison (setting III)}
\label{fig:MAEFixTaskIC}
\end{minipage}
\begin{minipage}{.49\linewidth}
\centering
\includegraphics[width=.98\textwidth]{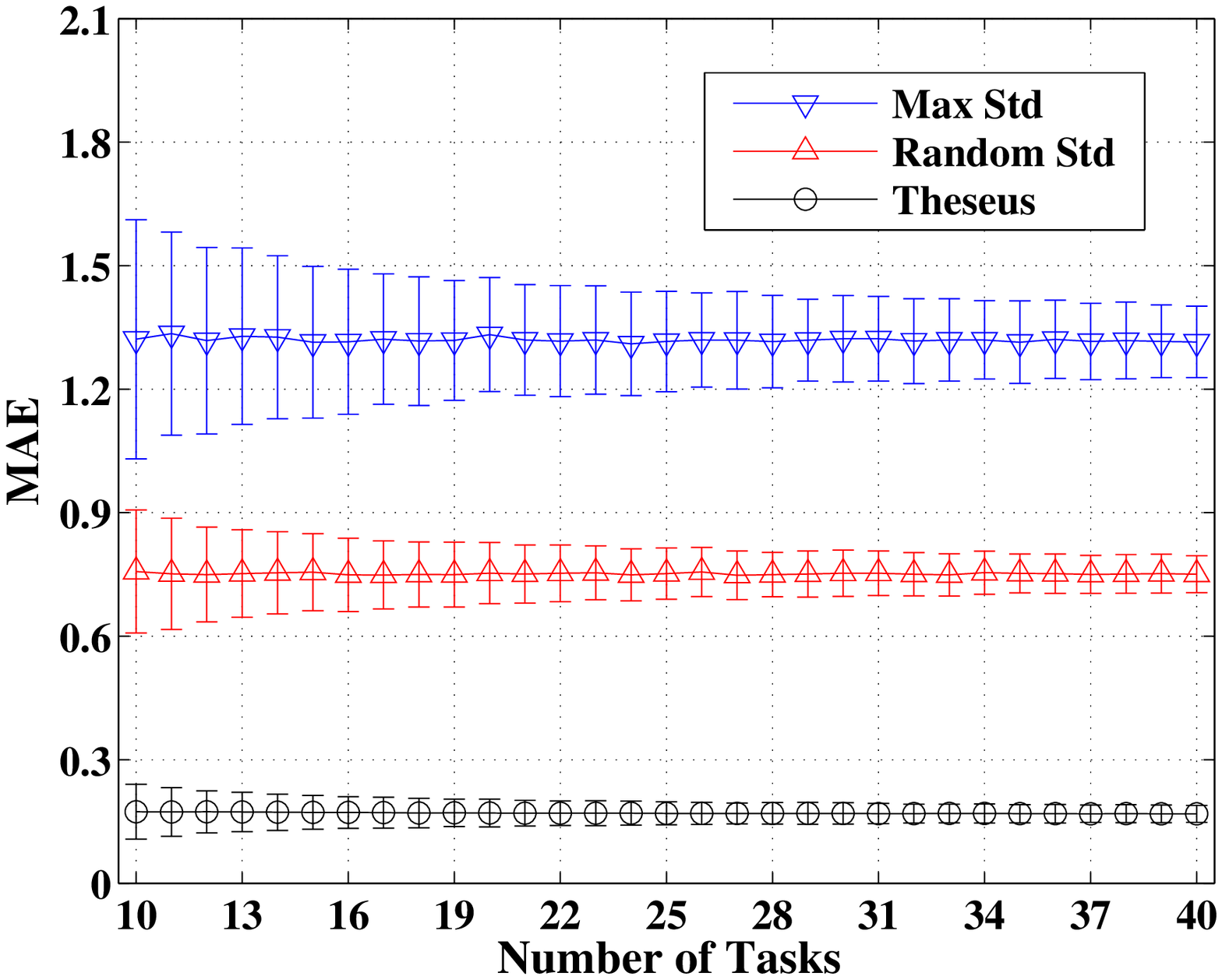}\\
\vspace{-0.45cm}
\caption{MAE comparison (setting IV)}
\label{fig:MAEFixWorkerIC}
\end{minipage}
\end{figure}



\section{Conclusion}

In conclusion, in this paper, we propose a payment mechanism, called Theseus, which is used in pair with a truth discovery algorithm to ensure high aggregation accuracy in MCS systems where workers may strategically reduce their sensing effort. Theseus tackles workers' strategic behavior, and incentivizes workers to spend their maximum possible effort at the BNE of the induced sensing game among workers. Furthermore, we ensure that Theseus bears other desirable properties, including individual rationality and budget feasibility. The desirable properties of Theseus are validated through theoretical analysis, and extensive simulations. 

\bibliographystyle{IEEEtran}
\bibliography{reference}
\end{document}